\documentclass[11pt]{article}
\usepackage[margin=1in]{geometry}
\usepackage{mathtools}
\usepackage{amssymb}
\usepackage{amsthm}
\usepackage{algorithm}
\usepackage{algpseudocode}
\usepackage{enumitem}
\usepackage{booktabs}
\usepackage{tikz}
\usetikzlibrary{backgrounds}
\usepackage{xcolor}
\usepackage{url}
\usepackage{hyperref}
\definecolor{csdsBlue}{RGB}{0,90,200}
\definecolor{csdsOrange}{RGB}{230,120,0}
\setlist{itemsep=0.2em,topsep=0.35em}

\theoremstyle{plain}
\newtheorem{theorem}{Theorem}[section]
\newtheorem{lemma}[theorem]{Lemma}
\newtheorem{proposition}[theorem]{Proposition}
\newtheorem{corollary}[theorem]{Corollary}
\newcommand{\gcs}{\gamma_{\mathrm{cs}}}
\newcommand{\priv}{\operatorname{pn}}
\newcommand{\CSDS}{\textnormal{\textsc{Co-Secure Dominating Set}}}
\providecommand{\doi}[1]{\url{https://doi.org/#1}}
\newcommand{\keywords}[1]{\par\medskip\noindent\textbf{Keywords: }#1\par}

\title{Co-Secure Domination in Split Graphs and Their Extensions: Complexity, Kernelization, and Approximation}

\author{M. S. Manjusha$^{1,*}$ and Pazhaniappan Renjith$^1$}
\date{}

\begin{document}
\maketitle
\begin{center}\small
$^1$ National Institute of Technology Calicut,\par
Corresponding author: manjusha\_p220226cs@nitc.ac.in\par
Email: renjith@nitc.ac.in\par
\end{center}
\begin{abstract}
Co-secure domination is a fault-tolerant form of domination: every selected vertex must have an adjacent unselected replacement whose insertion preserves domination. Split graphs are chordal, and co-secure domination is NP-complete on them. However, this hardness does not extend uniformly to more structured chordal graph classes. For connected chordal extensions of diameter at most five, we give a linear-time algorithm for graphs formed by adding positive-leaf pendant stars to a central clique, but prove NP-completeness for graphs formed by attaching peripheral cliques of any fixed size with unrestricted attachments. Thus, chordality and bounded diameter alone do not determine complexity.
Within split graphs, we establish a forbidden-star dichotomy. A minimum solution is computable in linear time on claw-free split graphs, whereas the decision problem is NP-complete on $K_{1,r}$-free split graphs for every fixed $r\geq4$, even when every clique vertex has exactly two neighbours in the independent set. For fixed $r$, this class has an $O(b^{r-1})$-vertex kernel parameterized by the solution budget $b$ and a linear-time $(r-1)$-approximation. A quadratic kernel applies when every clique vertex has at most two neighbours in the independent set.
We also classify pendant-star and peripheral-clique extensions. Certain proper attachment cases admit linear-time exact algorithms, while two attachment vertices with unequal central neighbourhoods already give NP-completeness for peripheral triangles. Proper attachments of bounded peripheral size admit a PTAS, but an FPTAS would imply $P=NP$ even on strict $K_3$-split graphs with exactly two attachment vertices in each peripheral triangle. This result does not establish APX-hardness or rule out a PTAS.
\end{abstract}
\par\medskip

\section{Introduction}\label{sec:introduction}

A dominating set may depend critically on a selected vertex: if that vertex becomes unavailable, some vertices can immediately lose coverage. Co-secure domination incorporates a local recovery requirement by asking that every selected vertex have an adjacent substitute whose insertion preserves domination. This makes the problem a natural fault-tolerant variant of domination, but it also couples two tasks that are usually studied separately---covering the graph and certifying a valid replacement for every chosen vertex.

For a graph $G=(V,E)$, a set $D\subseteq V$ is a \emph{co-secure dominating set} (CSDS) if it dominates $G$ and, for every $u\in D$, there is a vertex $v\in N_G(u)\setminus D$ such that
   $(D\setminus\{u\})\cup\{v\}$
also dominates $G$. We call $v$ a \emph{replacement} for $u$ with respect to $D$. Different selected vertices may have the same replacement; replacements are considered one at a time. The minimum cardinality of a CSDS is the \emph{co-secure domination number}, denoted by $\gcs(G)$. The optimisation problem asks for a minimum CSDS. Its decision version, denoted by \CSDS, asks whether $\gcs(G)\leq b$ for a given integer $b$.

Arumugam, Ebadi, and Manrique introduced co-secure domination and established basic bounds and complexity results~\cite{co-secureDS}. Kusum and Pandey proved NP-completeness on split graphs~\cite{KusumComplexity2023}; ordinary Dominating Set is also NP-complete on split graphs~\cite{Bertossi1984}. Since every split graph is chordal, these results immediately imply hardness on the broad class of chordal graphs. This inheritance can make it tempting to expect hardness to persist in more structured chordal graph families, leaving the boundary between split graphs and chordal graphs less explored. Yet natural extensions of split structure can behave quite differently. In the positive-leaf $K_{1,r}$-split graphs studied here, every independent centre has at least one pendant leaf, and co-secure domination is solvable exactly in linear time. By contrast, it is NP-complete on $K_s$-split graphs with unrestricted attachments, for every fixed $s$. Both families are chordal and, for connected instances with the stated partitions, have diameter at most five: every vertex is within distance two of the central clique, whose vertices are pairwise adjacent. Thus bounded diameter and chordal structure alone do not predict the complexity boundary. Here $K_{1,r}$-split refers to the pendant-star extension defined in Section~\ref{sec:extensions}, not to the $K_{1,r}$-free split graphs studied in Sections~\ref{sec:clawfree}--\ref{sec:approximation}. Kusum and Pandey also obtained hardness results for doubly chordal graphs and several subclasses of bipartite graphs, together with algorithms for chain graphs and graphs of bounded treewidth or clique-width~\cite{CSDD_split_chordalbipartite_convex}.

Related co-secure domination work includes a linear-time algorithm for proper interval graphs~\cite{ZouProperInterval2022}. A similar structural contrast appears in perfect Italian domination, a distinct domination variant: recent work gives polynomial-time algorithms on convex and triad-convex bipartite graphs, but proves NP-completeness on star-convex bipartite graphs~\cite{Renjith2026PID}. These results provide a related example of how graph structure can change the complexity of a domination problem.
Karthika et al.~\cite{KarthikaFCT} studied structural parameterizations of co-secure domination, including treewidth, modular-width, vertex cover, and distance to cluster or co-cluster graphs. Their work also records W[2]-hardness with respect to solution size on general graphs and excludes a polynomial kernel for the vertex-cover parameter unless $\mathsf{coNP}\subseteq\mathsf{NP/poly}$. Our parameterized analysis instead concerns the solution budget on split graphs with bounded clique-to-independent degree. To the best of our knowledge, this gives the first polynomial kernel for \CSDS\ parameterized by the solution size on this restricted split-graph setting. In particular, it yields polynomial kernels on $K_{1,r}$-free split graphs for every fixed $r$. These claims are class-specific and do not assert a polynomial kernel for general graphs.

Approximation results likewise depend on the graph class. Panda, Rana, and Mishra~\cite{Panda2024} gave an $O(\log |V|)$-approximation for graphs of minimum degree at least two and proved logarithmic inapproximability on specified bipartite subclasses. Wang, Yang, and Xu~\cite{Wang2025} obtained a $2(t-1)$-approximation for several geometric intersection classes that are $K_{1,t}$-free. These statements do not by themselves provide approximation guarantees or hardness bounds for all $K_{1,r}$-free split graphs. We give direct proofs for the split-graph classes considered here.
A corrigendum to the inapproximability paper is now available~\cite{PandaCorrigendum2026}.

We focus on split graphs, whose vertices can be partitioned into a clique and an independent set. The clique makes coverage within one part immediate, while neighbourhoods in the independent set determine the replacement constraints. Forbidding an induced star restricts these neighbourhoods and gives a natural hierarchy in which to locate a complexity boundary. A related classification for Steiner Tree places its boundary between $K_{1,4}$-free and $K_{1,5}$-free split graphs~\cite{RENJITH2020246}. We show that co-secure domination becomes hard already in the $K_{1,4}$-free class.

This setting leads to four questions. Where is the exact complexity boundary in the forbidden-star hierarchy? Can the structural restriction still yield a polynomial kernel when exact computation is NP-hard? Does it give a constant-factor approximation on the entire class? Finally, can the special incidence structure used in the hardness reduction support a better approximation? The following list summarizes our answers; the later sections explain the different structural ideas behind them.

\paragraph{Main results.}
For connected split graphs of order at least two, the following statements hold.
\begin{enumerate}[label=(\roman*),nosep]
\item A minimum CSDS can be found in $O(|V|+|E|)$ time in claw-free split graphs, given a split partition.
\item For every fixed integer $r\geq4$, \CSDS\ is NP-complete in $K_{1,r}$-free split graphs. Hardness holds even when every clique vertex has exactly two neighbours in the independent set.
\item For every fixed $r\geq4$, the decision problem on $K_{1,r}$-free split graphs has an $O(b^{r-1})$-vertex kernel, parameterized by the solution budget $b$. If every clique vertex has at most two independent neighbours, the kernel has at most $4b^2+4b+2$ vertices.
\item For every fixed $r\geq4$, a CSDS within a factor $r-1$ of optimum can be found in linear time in $K_{1,r}$-free split graphs, given a split partition.
\item If every $c\in C$ has exactly two neighbours in $I$, these neighbourhoods are pairwise distinct, and every vertex of $I$ has degree at least two, a CSDS within a factor $3/2$ of optimum can be found in linear time, given the split partition.
\end{enumerate}

The two sides of the dichotomy have different structural explanations. In the claw-free case, either the neighbourhoods of independent-set vertices are disjoint or the independent set has at most three vertices. These alternatives yield an explicit description of an optimum solution. For hardness, we first move a sufficiently small solution into the clique. We then interpret its vertices as edges of an auxiliary graph. Co-security excludes isolated-edge components, so a solution that meets the counting lower bound must encode a partition into three-vertex paths. This equality condition supplies the reverse implication of the reduction.

The same neighbourhood restrictions also support parameterized and approximation algorithms. A solution of size at most $b$ can cover at most $\Delta_I b$ independent vertices. Retaining two clique vertices for each independent-set neighbourhood then gives a polynomial kernel when $\Delta_I$ is fixed. A directly constructed independent CSDS yields a constant-factor approximation. For the incidence-split class used in the hardness proof, spanning forests and the component lower bound improve the guarantee to $3/2$.

We also study what happens when independent vertices are replaced by small connected graphs. Replacing them by disjoint cliques leads to bounded-peripheral-size unipolar graphs, a class studied by Eschen and Wang~\cite{EschenWang2014}. We distinguish peripheral size from the number and neighbourhoods of vertices attached to the central clique. One attachment vertex per peripheral clique yields an exact formula, whereas two attachment vertices with different central neighbourhoods already suffice for hardness. Pendant-star extensions give a complementary boundary: requiring a positive number of leaves at every centre makes the problem tractable.

Section~\ref{sec:preliminaries} introduces notation and a replacement criterion. Sections~\ref{sec:clawfree} and~\ref{sec:hardness} establish the dichotomy. Sections~\ref{sec:parameterized} and~\ref{sec:approximation} give the kernelization and approximation results. Section~\ref{sec:extensions} classifies pendant-star and peripheral-clique extensions, including partial attachments. Section~\ref{sec:conclusion} concludes the paper.

\section{Graph Theoretic Preliminaries}\label{sec:preliminaries}

All graphs are finite, simple, and undirected. Unless stated otherwise, an input graph is connected and has at least two vertices. We use the terminology of West~\cite{west2001introduction}. For $u\in V(G)$, let $N_G(u)$ and $N_G[u]=N_G(u)\cup\{u\}$ be its open and closed neighbourhoods. For $S\subseteq V(G)$, write $N_G[S]=\bigcup_{u\in S}N_G[u]$. Thus $D$ dominates $G$ if and only if $N_G[D]=V(G)$. We omit the subscript $G$ when the graph is clear.

A \emph{split partition} is a partition $V(G)=C\mathbin{\dot\cup}I$ in which $C$ is a clique and $I$ is independent. A maximal clique is inclusion-wise maximal; a maximum clique has largest cardinality. We require maximality only where stated. For $c\in C$, put
\[
 N_I(c)=N(c)\cap I,\qquad d_I(c)=|N_I(c)|,\qquad
 \Delta_I=\max_{c\in C}d_I(c).
\]
A graph is \emph{$K_{1,r}$-free} if it has no induced copy of the star with $r$ leaves. The star $K_{1,3}$ is a \emph{claw}.

For a dominating set $D$ and a vertex $u\in D$, define the private neighbourhood
\[
 \priv_D(u)=\{x\in V(G):N[x]\cap D=\{u\}\}.
\]
These are precisely the vertices that lose domination when $u$ is removed. The following elementary criterion makes the replacement requirement explicit.

\begin{lemma}[Replacement criterion]\label{lem:replacement}
Let $D$ be a dominating set, $u\in D$, and $v\in N(u)\setminus D$. Then $v$ replaces $u$ if and only if
\[
                    \priv_D(u)\subseteq N[v].
\]
\end{lemma}
\begin{proof}
Every vertex outside $\priv_D(u)$ is dominated by $D\setminus\{u\}$. The remaining vertices must be dominated by $v$.
\end{proof}

Two simple consequences will be used repeatedly. First, the smallest possible solution illustrates how restrictive the replacement condition can be: domination by one vertex is insufficient unless a second vertex can take over the same role. Second, a general upper bound in terms of the independent side of a split partition provides a baseline against which the later exact and approximation results can be compared.

\begin{lemma}\label{lem:singleton}
A graph has co-secure domination number one if and only if it has at least two universal vertices.
\end{lemma}
\begin{proof}
If $D=\{u\}$ is a CSDS and $v$ replaces $u$, then both $u$ and $v$ are universal. Conversely, two universal vertices replace each other in singleton dominating sets.
\end{proof}

The next proposition records the baseline bound. Although it follows from a known independence-number bound, the split-graph formulation is useful because all subsequent algorithms are expressed in terms of the partition $(C,I)$.

\begin{proposition}\label{prop:splitbound}
Every non-trivial connected split graph with partition $(C,I)$ satisfies
\[
                         \gcs(G)\leq |I|+1.
\]
The bound is sharp.
\end{proposition}
\begin{proof}
Arumugam et al.~\cite{co-secureDS} proved that $\gcs(G)\leq\alpha(G)$, where $\alpha(G)$ is the independence number. An independent set contains at most one vertex of $C$, so $\alpha(G)\leq|I|+1$.

For sharpness, take $G=K_{1,t}$ with $t\geq2$, and let $C$ consist of the centre and one leaf. A CSDS cannot contain the centre: no selected leaf would then have an unselected neighbour, and the centre alone has no valid replacement. Hence every leaf must be selected. The leaves form a CSDS, giving $\gcs(G)=t=|I|+1$.
\end{proof}

\section{Claw-Free Split Graphs}\label{sec:clawfree}

The first algorithmic question is whether forbidding a claw constrains the interaction between the clique and the independent set strongly enough to determine an optimum CSDS. A claw centred at a clique vertex can use either three independent neighbours, or two independent neighbours together with a suitably chosen clique vertex. Consequently, claw-freeness restricts both the size of each $N_I(c)$ and the way these neighbourhoods intersect. The following lemma makes these restrictions explicit; it is the structural reason that only a small number of cases remain in the exact algorithm. We include its proof to make the algorithm self-contained~\cite{RENJITH2020246}.

\begin{lemma}\label{lem:clawstructure}
Let $G$ be a connected split graph with split partition $(C,I)$. Then $G$ is claw-free if and only if either $\Delta_I\leq1$, or $\Delta_I=2$ and
\begin{equation}\label{eq:intersection}
 N_I(c)\cap N_I(d)\neq\varnothing
 \quad\text{whenever }c,d\in C\text{ and }d_I(c)=2.
\end{equation}
If $G$ is claw-free, $C$ is maximal, and $\Delta_I=2$, then $|I|\leq3$.
\end{lemma}
\begin{proof}
A claw cannot be centred in $I$, because its neighbours form a clique. A clique vertex with three neighbours in $I$ centres a claw. If $N_I(c)=\{x,y\}$ and $N_I(d)\cap\{x,y\}=\varnothing$, then $c$ centres the claw with leaves $x,y,d$. Conversely, a claw centred in $C$ has either three leaves in $I$, or two leaves in $I$ and one in $C$. The stated conditions exclude both possibilities.

Now assume that $G$ is claw-free, $C$ is maximal, $\Delta_I=2$, and $|I|\geq4$. The two-element sets among the neighbourhoods $N_I(c)$ are pairwise intersecting. Such a family either has a common element or consists of subsets of a fixed three-element set. To see this, take two distinct members $\{a,b\}$ and $\{a,d\}$. Any member avoiding $a$ must be $\{b,d\}$, after which every member lies in $\{a,b,d\}$.

In the three-element alternative with no common element, no singleton neighbourhood intersects all members, so connectivity gives $|I|\leq3$, a contradiction. In the common-element alternative, the assumption $|I|\geq4$ forces at least three distinct two-element neighbourhoods. Every singleton neighbourhood must then consist of their common element. By~\eqref{eq:intersection}, no clique vertex has an empty neighbourhood in $I$. The common element is therefore adjacent to all of $C$, contradicting maximality of $C$.
\end{proof}

Lemma~\ref{lem:clawstructure} suggests a direct case analysis. When $\Delta_I=1$, the independent vertices impose disjoint domination requirements. When $\Delta_I=2$, maximality of the clique forces $|I|\leq3$, so the remaining possibilities can be inspected explicitly. The next theorem turns this structural dichotomy into exact values and, equally importantly, supplies a replacement witness for every selected vertex.

\begin{theorem}\label{thm:clawvalues}
Let $G$ be a connected claw-free split graph of order at least two, with a split partition $(C,I)$ in which $C$ is maximal. Its co-secure domination number is as follows:
\[
\gcs(G)=
\begin{cases}
1,& I=\varnothing;\\
2,& \Delta_I=1,\ I=\{u\},\ d(u)=1;\\
1,& \Delta_I=1,\ I=\{u\},\ d(u)\geq2;\\
|I|,& \Delta_I=1,\ |I|\geq2;\\
1,& \Delta_I=2,\ |I|=2,\text{ and at least two }c\in C
       \text{ have }d_I(c)=2;\\
2,& \Delta_I=2\text{ in all other cases}.
\end{cases}
\]
An optimum set and a replacement for each selected vertex can be found in $O(|V|+|E|)$ time.
\end{theorem}

Algorithm~\ref{alg:clawfree} gives the construction proved below.

\begin{algorithm}[!t]
\caption{Minimum CSDS in a claw-free split graph}\label{alg:clawfree}
\begin{algorithmic}[1]
\Require Connected claw-free split graph $G$ of order at least two; maximal split partition $(C,I)$
\Ensure A minimum co-secure dominating set
\If{$I=\varnothing$}
  \State \Return $\{c\}$ for any $c\in C$
\EndIf
\State Compute $N_I(c)$ for $c\in C$ and $d(u)$ for $u\in I$
\If{$\Delta_I=1$}
  \If{$|I|\geq2$}
    \State Choose $c_u\in N(u)$ for each $u\in I$
    \State \Return $\{c_u:u\in I\}$
  \EndIf
  \State Let $I=\{u\}$
  \If{$d(u)\geq2$}
    \State \Return $\{c\}$ for any $c\in N(u)$
  \EndIf
  \State \Return $\{u,x\}$ for any $x\in C\setminus N(u)$
\EndIf
\State Choose $c\in C$ with $d_I(c)=2$
\If{$|I|=2$}
  \If{there is $d\in C\setminus\{c\}$ with $d_I(d)=2$}
    \State \Return $\{c\}$
  \EndIf
  \State \Return $\{c,d\}$ for any $d\in C\setminus\{c\}$
\EndIf
\State Let $\{w\}=I\setminus N_I(c)$ and choose $d\in N(w)$
\State \Return $\{c,d\}$
\end{algorithmic}
\end{algorithm}

\begin{proof}
We follow the alternatives in Lemma~\ref{lem:clawstructure}. In each case, we first derive a lower bound from domination or from Lemma~\ref{lem:singleton}, and then construct a CSDS of matching size together with explicit replacements.

If $I=\varnothing$, then $G$ is complete and Lemma~\ref{lem:singleton} applies. Connectivity excludes $\Delta_I=0$ when $I\neq\varnothing$.

Suppose $\Delta_I=1$. Distinct vertices of $I$ have disjoint neighbourhoods in $C$. If $|I|\geq2$, every dominating set must meet each of the pairwise disjoint sets $N[u]$, $u\in I$, so it has size at least $|I|$. Choose one neighbour $c_u\in C$ of each $u\in I$. The vertices $c_u$ are distinct, and
\[
                         D=\{c_u:u\in I\}
\]
is a CSDS: $u$ replaces $c_u$, while another selected clique vertex continues to dominate $C$.

If $I=\{u\}$ and $d(u)\geq2$, any two neighbours of $u$ are universal, giving an optimum singleton set. If $d(u)=1$, let $c$ be its neighbour. Maximality of $C$ implies that $C$ contains a vertex $x\neq c$ not adjacent to $u$. The set $\{u,x\}$ is co-secure, with $c$ replacing either selected vertex. There are not two universal vertices, so its size is optimal.

Now suppose $\Delta_I=2$. Lemma~\ref{lem:clawstructure} gives $|I|\in\{2,3\}$. If $I=\{u,v\}$, the clique vertices adjacent to both $u$ and $v$ are universal. Two such vertices give an optimum singleton set. If there is exactly one, say $c$, choose $d\in C\setminus\{c\}$. By~\eqref{eq:intersection}, $d$ is adjacent to exactly one vertex of $I$, say $u$. Then $D=\{c,d\}$ is co-secure: $v$ replaces $c$, and $u$ replaces $d$. Lemma~\ref{lem:singleton} proves optimality.

Finally, let $I=\{u,v,w\}$. Choose $c\in C$ with $N_I(c)=\{u,v\}$ and a neighbour $d\in C$ of $w$. Equation~\eqref{eq:intersection} implies that $d$ is adjacent to one of $u,v$, say $u$. Thus $D=\{c,d\}$ is co-secure, with replacements $v$ and $w$, respectively. No vertex is universal, so $|D|=2$ is optimal.

For the running time, scan the adjacency lists to compute $N_I(c)$, $d_I(c)$, and the degrees of vertices in $I$. The case $\Delta_I=1$ requires at most one selected neighbour per independent vertex. When $\Delta_I=2$, at most three independent vertices occur, and the required clique vertices can be found by a scan. No pairwise search over $C$ is needed. Each step takes $O(|V|+|E|)$ time.
\end{proof}

If a supplied split partition does not have a maximal clique, move a vertex of $I$ adjacent to all of $C$ into $C$. At most one such move is needed, since the original independent set has no edges. The test and the move take linear time. Thus Theorem~\ref{thm:clawvalues} also applies when an arbitrary split partition is supplied.

\section{Hardness in \texorpdfstring{$K_{1,4}$}{K1,4}-Free Split Graphs}\label{sec:hardness}

The claw-free algorithm raises the natural boundary question: does tractability persist when one more independent leaf is allowed? To prove that it does not, the reduction must control both parts of the CSDS definition. Domination is a covering condition, whereas co-security is a local replacement condition, and a direct set-cover reduction does not automatically encode the latter.

Our strategy separates these two requirements through an auxiliary graph. Every clique vertex represents an edge, and its two independent neighbours represent the endpoints of that edge. Selecting clique vertices therefore produces an edge cover. The replacement condition has a precise additional effect: it forbids a selected edge from forming an isolated $K_2$ component. A tight counting bound then forces every component to be a $P_3$, which is exactly the structure needed to recover a $P_3$-partition of the source instance.

\subsection{Moving Small Solutions into the Clique}

The edge interpretation is available only for solutions contained in the clique. An arbitrary CSDS may also select independent vertices, so the first obstacle is to show that a sufficiently small solution can be normalised without increasing its size. This is delicate because replacing a selected independent vertex by a clique neighbour can destroy replacement witnesses elsewhere. The next lemma performs the exchange while controlling all private neighbourhoods; it does not assume that an arbitrary exchange preserves co-security.

\begin{lemma}[Clique normalisation]\label{lem:normalisation}
Let $G$ have split partition $(C,I)$, and suppose every vertex of $I$ has degree at least two. If $D$ is a CSDS with $D\cap C\neq\varnothing$ and $|D|<|C|$, then $G$ has a CSDS $D_C\subseteq C$ with $|D_C|\leq|D|$.
\end{lemma}
\begin{proof}
The proof repeatedly removes one selected independent vertex. Along with that vertex, we discard any other selected independent vertices that become redundantly dominated by two clique vertices. This redundancy is what preserves co-security for the remaining solution.

Write $A=D\cap C$ and $S=D\cap I$. If $S=\varnothing$, there is nothing to prove. Otherwise, choose $x\in S$. Co-security gives a neighbour $c\in C\setminus A$ of $x$. Put
\begin{align*}
 A'&=A\cup\{c\},\qquad
 Z=\{z\in S\setminus\{x\}:N(z)\subseteq A'\},\\
 D'&=A'\cup(S\setminus(\{x\}\cup Z)).
\end{align*}
The set $D'$ has size at most $|D|$ and contains fewer vertices of $I$. It dominates $G$: vertices originally outside $S$ remain covered by $A$, the vertex $x$ is covered by $c$, and every vertex of $Z$ has at least two neighbours in $A'$.

Since $|A'|\geq2$, no clique vertex belongs to the private neighbourhood of a selected vertex. The only possible private independent vertex of $c$ is $x$, so $x$ replaces $c$ by Lemma~\ref{lem:replacement}. For each remaining vertex $z\in D'\cap I$, its private neighbourhood is contained in $\{z\}$. The definition of $Z$ supplies a neighbour of $z$ outside $A'$, which replaces $z$.

Consider $a\in A$ and an original replacement $r$ for $a$. No new private independent vertex is created for $a$: the vertex $x$ is covered by $c$, and each vertex of $Z$ has two selected neighbours. If $r\neq c$, it remains outside $D'$ and still covers $\priv_{D'}(a)$. If $r=c$, every originally private independent vertex of $a$ is now covered by $c$, so $\priv_{D'}(a)=\varnothing$. In this case any vertex of $C\setminus D'$ replaces $a$; such a vertex exists because $|D'|<|C|$. Hence $D'$ is co-secure. Repeating the exchange removes all selected vertices of $I$.
\end{proof}

\subsection{An Auxiliary Edge-Cover Formulation}

After normalisation, the selected vertices lie in the clique. We now choose a split-graph representation in which such a selection has a familiar combinatorial meaning. The incidence construction below turns domination of the independent set into edge coverage in an auxiliary graph and allows the replacement obstruction to be read from component structure.

For a simple graph $J=(W,F_J)$, define a split graph $S(J)$ by
\[
 I=W,\qquad C=\{c_e:e\in F_J\}.
\]
Make $C$ a clique, and join $c_e$ to the two endpoints of $e$. These are the only edges between $C$ and $I$. An \emph{edge cover} of $J$ is a set of edges incident with every vertex of $J$. For $D\subseteq C$, write
\[
                  F_D=\{e\in F_J:c_e\in D\}.
\]

\begin{lemma}\label{lem:edgecover}
Let $D\subsetneq C$ with $|D|\geq2$. Then $D$ is a CSDS of $S(J)$ if and only if $F_D$ is an edge cover of $J$ and the spanning subgraph $(W,F_D)$ has no component isomorphic to $K_2$.
\end{lemma}
\begin{proof}
There are two conditions to verify. Domination determines whether the chosen auxiliary edges cover all vertices, while the private endpoints of each chosen edge determine whether its clique representative has a valid replacement.

The set $D$ dominates $C$. It dominates $I$ exactly when $F_D$ is an edge cover. Let $e=xy\in F_D$. An endpoint of $e$ is private to $c_e$ exactly when it has degree one in $(W,F_D)$.

If neither endpoint is private, any vertex of $C\setminus D$ replaces $c_e$. If exactly one endpoint is private, that endpoint replaces $c_e$; another selected clique vertex continues to dominate $C$. If both endpoints are private, they form a $K_2$ component of $(W,F_D)$. No independent vertex dominates both endpoints. No other clique vertex does so either, because $J$ is simple and $e$ is its unique edge with endpoints $x,y$. Thus $c_e$ has a replacement exactly when $e$ is not an isolated-edge component.
\end{proof}

Lemma~\ref{lem:edgecover} reduces the optimisation question to finding a small spanning subgraph with neither isolated vertices nor isolated-edge components. The next bound identifies the smallest possible density of such a subgraph. Its equality case is crucial: the reduction budget will be chosen so that every feasible solution must attain equality.

\begin{lemma}[Counting and equality]\label{lem:counting}
If a spanning subgraph $(W,F')$ has no isolated vertices and no $K_2$ components, then
\[
                            |F'|\geq\frac{2|W|}{3}.
\]
Equality holds if and only if every component is a three-vertex path.
\end{lemma}
\begin{proof}
A component on $s\geq3$ vertices has at least $s-1\geq2s/3$ edges. Equality in both inequalities requires $s=3$ and exactly two edges. Summing over components proves the claim.
\end{proof}

\subsection{The Path Gadget and Its Equality States}

A \emph{$P_3$-partition} is a partition of the vertex set into triples, each containing a path with two edges. The paths need not be induced. Deciding whether a graph has a $P_3$-partition is NP-complete, even for bipartite graphs of maximum degree three~\cite{MonnotToulouse2007}. We may restrict to graphs with no isolated vertices: an isolated vertex certifies a no-instance and can be handled before the reduction.

The counting lemma indicates that $P_3$-partition is the appropriate source problem, but applying the incidence construction directly would not ensure that every independent vertex has degree at least two, a condition needed by the normalisation lemma. The following replacement gadget raises all relevant degrees while preserving $P_3$-partitionability. Its three local states also make the reverse direction of the reduction transparent.

Let $H=(U,F)$ be such an instance, with $|U|=3p$, $p\geq1$, and $m=|F|$. Form $J(H)$ by replacing each edge $uv\in F$ with two internally vertex-disjoint paths of four edges. All internal vertices are new and are distinct across gadgets. Thus
\begin{equation}\label{eq:gadgetsizes}
                  |V(J(H))|=3p+6m,\qquad |E(J(H))|=8m.
\end{equation}
Every original vertex $u$ has degree $2d_H(u)\geq2$ in $J(H)$, and every new vertex has degree two. Figure~\ref{fig:gadget} shows one gadget.

\begin{figure}[t]
\centering
\begin{tikzpicture}[every node/.style={font=\small},
  vertex/.style={circle,draw,fill=white,inner sep=2pt}]
\node[vertex,label=left:$u$] (start) at (0,0) {};
\node[vertex,label=right:$v$] (end) at (6,0) {};
\node[vertex,label=above:$a_1$] (upperfirst) at (1.5,0.8) {};
\node[vertex,label=above:$a_2$] (uppermiddle) at (3,0.8) {};
\node[vertex,label=above:$a_3$] (upperlast) at (4.5,0.8) {};
\node[vertex,label=below:$b_1$] (lowerfirst) at (1.5,-0.8) {};
\node[vertex,label=below:$b_2$] (lowermiddle) at (3,-0.8) {};
\node[vertex,label=below:$b_3$] (lowerlast) at (4.5,-0.8) {};
\draw (start)--(upperfirst)--(uppermiddle)--(upperlast)--(end);
\draw (start)--(lowerfirst)--(lowermiddle)--(lowerlast)--(end);
\end{tikzpicture}
\caption{The gadget in $J(H)$ for an edge $uv$ of $H$. Each displayed edge corresponds to one clique vertex in $S(J(H))$.}
\label{fig:gadget}
\end{figure}

The next lemma verifies that the gadget changes the local representation of a path partition without changing whether one exists. The proof classifies all possible ways in which a $P_3$-partition can traverse one four-edge path; these local states then assemble globally.

\begin{lemma}\label{lem:pathpartition}
The graph $H$ has a $P_3$-partition if and only if $J(H)$ has a $P_3$-partition.
\end{lemma}
\begin{proof}
We first determine the admissible local states of one gadget path. We then show that active states correspond exactly to oriented edges of a $P_3$-partition of the original graph.

Consider a gadget path $Q=(u,a_1,a_2,a_3,v)$. Represent a chosen subset of its four edges by a binary word in path order. In a spanning collection of $P_3$ components, the only possible words are
\[
                         0110,\qquad 1101,\qquad 1011.
\]
Indeed, if $a_2$ is the centre of its component, the word is $0110$. If $a_2$ is a leaf adjacent to $a_1$, then $a_1$ must also be adjacent to $u$, and $a_3$ must be adjacent to $v$, giving $1101$. The case where $a_2$ is adjacent to $a_3$ gives $1011$.

Call $0110$ \emph{inactive}. It covers the three internal vertices without using either endpoint. Each other state is \emph{active}: one original endpoint is a leaf in a completed path, while the other is a centre that must be joined to one more internal vertex from another active path. Table~\ref{tab:states} records these roles.

\begin{table}[t]
\centering
\setlength{\tabcolsep}{5pt}
\caption{The three states of $Q=(u,a_1,a_2,a_3,v)$ in a $P_3$-partition.}
\label{tab:states}
\begin{tabular}{cccc}
\toprule
State & Selected edges & Role of $u$ & Role of $v$\\
\midrule
$0110$ & $a_1a_2,\ a_2a_3$ & unused & unused\\
$1101$ & $ua_1,\ a_1a_2,\ a_3v$ & leaf & centre\\
$1011$ & $ua_1,\ a_2a_3,\ a_3v$ & centre & leaf\\
\bottomrule
\end{tabular}
\end{table}

Given a $P_3$-partition of $H$, orient its two edges in each path from the leaves towards the centre. For each oriented edge, activate one of its two gadget paths in the corresponding leaf-to-centre state. Leave every other gadget path inactive. Each original leaf lies in one completed $P_3$, and each original centre joins the two remaining internal leaves. Every internal vertex is covered exactly once, so this gives a $P_3$-partition of $J(H)$.

Conversely, consider a $P_3$-partition of $J(H)$. Every original vertex is either the leaf endpoint of exactly one active path or the centre endpoint of exactly two active paths. Replace each active path by its original edge, oriented from its leaf endpoint to its centre endpoint. The resulting components are three-vertex paths. The two leaves at a centre are distinct: otherwise one original vertex would be required to be the leaf endpoint of two active paths. Hence these paths partition $U$ and give a $P_3$-partition of $H$.
\end{proof}

\subsection{The Reduction}

We can now combine the three ingredients. Clique normalisation removes independent vertices from a small solution, Lemma~\ref{lem:edgecover} translates the resulting clique solution into an auxiliary edge cover, and Lemma~\ref{lem:counting} forces that edge cover to be a $P_3$-partition at the chosen budget. Lemma~\ref{lem:pathpartition} then transfers the partition back to the source graph. Figure~\ref{fig:split-reduction} illustrates how these layers fit together.

\begin{figure}[p]
\centering
\begin{tikzpicture}[
  vertex/.style={circle,fill=black,inner sep=0pt,minimum size=4pt},
  every label/.style={font=\small}]
\foreach \index/\horizontal/\vertical in {1/0/1.2,2/2/1.2,3/4/1.2,4/0/0,5/2/0,6/4/0}{
  \ifnum\index<4
    \node[vertex,label=above:$v_{\index}$] (source\index) at (\horizontal,\vertical) {};
  \else
    \node[vertex,label=below:$v_{\index}$] (source\index) at (\horizontal,\vertical) {};
  \fi
}
\foreach \first/\second in {1/2,2/3,4/5,5/6,1/4,2/5,3/6}
  \draw (source\first)--(source\second);
\end{tikzpicture}

\smallskip
{\small (a) The source graph $H$.\par}
\medskip

\resizebox{\linewidth}{!}{%
\begin{tikzpicture}[yscale=1.15,
  original/.style={circle,fill=black,inner sep=0pt,minimum size=6pt},
  vertex/.style={circle,fill=black,inner sep=0pt,minimum size=3pt},
  vertexlabel/.style={font=\footnotesize,inner sep=1pt,fill=white},
  frame/.style={draw,thin}]
\newcommand{\reductionvertexlabel}[3]{\ensuremath{v^{#1}_{#2#3}}}
\foreach \index/\horizontal/\vertical in {1/1.7/6,2/7.3/6,3/12.9/6,4/1.7/0.4,5/7.3/0.4,6/12.9/0.4}{
  \node[original] (original\index) at (\horizontal,\vertical) {};
  \node[vertexlabel,anchor=north west] at (\horizontal+0.06,\vertical-0.06) {$v_{\index}$};
}
\draw[frame] (2.08,6.64) rectangle (11.98,8.18);
\draw[frame] (-0.22,1.28) rectangle (0.72,5.7);
\foreach \shift in {0,5.6}{
  \draw[frame] (2.18+\shift,2.55)--(5.4+\shift,2.55)
    --(5.4+\shift,5.7)--(6.34+\shift,5.7)
    --(6.34+\shift,0.88)--(2.18+\shift,0.88)--cycle;
}
\foreach \first/\second/\originx/\originy/\direction in {
  1/2/1.7/6/0,2/3/7.3/6/0,4/5/1.7/0.4/0,5/6/7.3/0.4/0,
  1/4/1.7/6/1,2/5/7.3/6/1,3/6/12.9/6/1}{
  \foreach \reverse/\vertexindex/\distance in {
    1/7/0.8,0/1/1.2,1/5/1.6,0/3/2.0,
    1/3/2.8,0/5/3.2,1/1/3.6,0/7/4.0}{
    \ifnum\reverse=0
      \edef\fromindex{\first}\edef\toindex{\second}
    \else
      \edef\fromindex{\second}\edef\toindex{\first}
    \fi
    \ifnum\direction=0
      \node[vertex] (g\fromindex\toindex-\vertexindex)
        at (\originx+\distance,\originy+1.2) {};
      \ifnum\reverse=1
        \node[vertexlabel,anchor=south] at (\originx+\distance-0.08,\originy+1.38)
          {\reductionvertexlabel{\vertexindex}{\fromindex}{\toindex}};
      \else
        \node[vertexlabel,anchor=north] at (\originx+\distance,\originy+1.08)
          {\reductionvertexlabel{\vertexindex}{\fromindex}{\toindex}};
      \fi
    \else
      \node[vertex] (g\fromindex\toindex-\vertexindex)
        at (\originx-1.2,\originy-\distance) {};
      \node[vertexlabel,anchor=east] at (\originx-1.3,\originy-\distance+0.09)
        {\reductionvertexlabel{\vertexindex}{\fromindex}{\toindex}};
    \fi
  }
  \foreach \reverse/\vertexindex/\distance in {
    1/6/0.8,0/2/1.6,1/4/2.4,0/4/3.2,1/2/4.0,0/6/4.8}{
    \ifnum\reverse=0
      \edef\fromindex{\first}\edef\toindex{\second}
    \else
      \edef\fromindex{\second}\edef\toindex{\first}
    \fi
    \ifnum\direction=0
      \node[vertex] (g\fromindex\toindex-\vertexindex)
        at (\originx+\distance,\originy) {};
      \node[vertexlabel,anchor=north] at (\originx+\distance,\originy-0.05)
        {\reductionvertexlabel{\vertexindex}{\fromindex}{\toindex}};
    \else
      \node[vertex] (g\fromindex\toindex-\vertexindex)
        at (\originx,\originy-\distance) {};
      \node[vertexlabel,anchor=west] at (\originx+0.08,\originy-\distance+0.09)
        {\reductionvertexlabel{\vertexindex}{\fromindex}{\toindex}};
    \fi
  }
  \begin{scope}[on background layer]
  \foreach \fromindex/\toindex in {\first/\second,\second/\first}{
    \draw[thin] (original\fromindex)--(g\fromindex\toindex-1);
    \foreach \step [evaluate=\step as \nextstep using int(\step+1)] in {1,...,6}
      \draw[thin] (g\fromindex\toindex-\step)--(g\fromindex\toindex-\nextstep);
    \draw[thin] (g\fromindex\toindex-7)--(original\toindex);
  }
  \end{scope}
}
\foreach \selected/\replacement in {
  g12-1/original1,
  g12-3/g12-4,
  g12-7/g12-6,
  g21-3/g21-2,
  g21-5/g21-6,
  g23-1/g23-2,
  g23-5/g23-4,
  g23-7/original3,
  g32-3/g32-2,
  g32-5/g32-6,
  g45-1/original4,
  g45-3/g45-4,
  g45-7/g45-6,
  g54-3/g54-2,
  g54-5/g54-6,
  g56-1/g56-2,
  g56-5/g56-4,
  g56-7/original6,
  g65-3/g65-2,
  g65-5/g65-6,
  g14-3/g14-2,
  g14-5/g14-6,
  g41-3/g41-2,
  g41-5/g41-6,
  g25-3/g25-2,
  g25-5/g25-6,
  g52-3/g52-2,
  g52-5/g52-6,
  g36-3/g36-2,
  g36-5/g36-6,
  g63-3/g63-2,
  g63-5/g63-6}{
  \begin{scope}[on background layer]
    \draw[csdsOrange,line width=0.7pt] (\selected)--(\replacement);
  \end{scope}
  \node[circle,fill=csdsBlue,inner sep=0pt,minimum size=6pt] at (\selected.center) {};
  \node[circle,fill=csdsOrange,inner sep=0pt,minimum size=6pt] at (\replacement.center) {};
}
\end{tikzpicture}%
}

\smallskip
{\small (b) The split graph $G=S(J(H))$, with clique edges omitted.\par}
\smallskip
{\small \textcolor{csdsBlue}{$\bullet$ Selected vertex in $D$}\qquad
\textcolor{csdsOrange}{$\bullet$ Replacement vertex}\qquad
\textcolor{csdsOrange}{--- Replacement pair}\par}

\caption{The $K_{1,4}$-free split-graph reduction for a six-vertex instance $H$. Each edge $v_iv_j$ gives two alternating paths, $(v_i,v^1_{ij},\ldots,v^7_{ij},v_j)$ and $(v_j,v^1_{ji},\ldots,v^7_{ji},v_i)$. The original vertices and all even-indexed internal vertices form the independent set $I$. The odd-indexed vertices enclosed by the rectangle and L-shaped outlines collectively form a single clique $C$, not separate cliques; all edges within $C$ are omitted. Each clique vertex has exactly two neighbours in $I$. Here $p=2$ and $m=7$, so $|I|=48$, $|C|=56$, and the reduction budget is $b=32$. The blue vertices form a CSDS of size $32$; orange vertices are their replacements, and orange edges identify the pairs.}
\label{fig:split-reduction}
\end{figure}

\begin{theorem}\label{thm:hardness}
\CSDS\ is NP-complete on connected split graphs in which every clique vertex has exactly two neighbours in the independent set.
\end{theorem}
\begin{proof}
The budget is chosen to equal the lower bound in Lemma~\ref{lem:counting}. Thus a source partition yields a CSDS of exactly that size, while any CSDS within the budget must attain equality and therefore decompose into $P_3$ components.

A certificate consists of a set $D$ and a replacement for each vertex of $D$. Domination can be checked for $D$ and for every replacement set in polynomial time, so the problem belongs to NP.

For hardness, take the $P_3$-partition instance $H$ above, construct
\[
                         G=S(J(H)),\qquad b=4m+2p,
\]
and use the split partition $(C,I)$ from the definition of $S(J(H))$. By~\eqref{eq:gadgetsizes}, $|C|=8m$ and $|I|=6m+3p$. Every clique vertex has exactly two neighbours in $I$, and every independent vertex has degree at least two. The graph $G$ is connected because $C$ is a clique and every vertex of $I$ has a neighbour in $C$. The construction has $14m+3p$ vertices and $\binom{8m}{2}+16m$ edges, so it is polynomial in the size of $H$. Figure~\ref{fig:split-reduction} illustrates the construction for a six-vertex source graph. In the figure, the odd-indexed vertices represent edges of $J(H)$, while the even-indexed vertices represent its new internal vertices.

Since $H$ has no isolated vertices, $2m\geq3p$. Consequently,
\begin{equation}\label{eq:budget}
                 b<|C|,\qquad b<|I|,\qquad b=\frac{2|I|}{3}.
\end{equation}

Suppose first that $H$ has a $P_3$-partition. Lemma~\ref{lem:pathpartition} gives a $P_3$-partition of $J(H)$ with $p+2m$ components and $2p+4m=b$ selected edges. Their corresponding clique vertices form a CSDS of $G$ by Lemma~\ref{lem:edgecover}.

Conversely, suppose $G$ has a CSDS $D$ with $|D|\leq b$. If $D\cap C=\varnothing$, domination of $I$ forces $D=I$, contradicting~\eqref{eq:budget}. Lemma~\ref{lem:normalisation} therefore gives a CSDS $D_C\subseteq C$ with $|D_C|\leq b<|C|$. Since $|I|>2$ and each clique vertex covers only two independent vertices, $|D_C|\geq2$. Lemmas~\ref{lem:edgecover} and~\ref{lem:counting} imply
\[
               b\geq |D_C|=|F_{D_C}|\geq\frac{2|I|}{3}=b.
\]
Equality holds throughout. Thus every component of $(I,F_{D_C})$ is a $P_3$. Lemma~\ref{lem:pathpartition} yields a $P_3$-partition of $H$, completing the reduction.
\end{proof}

The constructed instances already satisfy the stronger forbidden-subgraph requirement needed for the dichotomy. The following corollary records this fact explicitly rather than relying only on the bound $d_I(c)=2$.

\begin{corollary}\label{cor:stars}
For every fixed integer $r\geq4$, \CSDS\ is NP-complete on connected $K_{1,r}$-free split graphs.
\end{corollary}
\begin{proof}
In the constructed graph, an induced star centred in $I$ has at most one leaf because its neighbours form a clique. A star centred in $C$ has at most two leaves in $I$ and at most one in $C$. Thus the graph is $K_{1,4}$-free and hence $K_{1,r}$-free for every $r\geq4$. The result follows from Theorem~\ref{thm:hardness} and membership in NP.
\end{proof}

Theorem~\ref{thm:clawvalues} and Corollary~\ref{cor:stars} establish the first two claims in the introduction's result summary. The argument establishes the forbidden-star property directly; it does not rely on composing reductions that are only known to preserve the class of split graphs.

\section{Solution-Size Parameterization}\label{sec:parameterized}

Let $b\geq1$ be the solution budget. An algorithm is \emph{fixed-parameter tractable} (FPT) with respect to $b$ if its running time is $f(b)(|V|+|E|)^{O(1)}$ for a computable function $f$. A \emph{kernelization} is a polynomial-time reduction to an equivalent instance whose size is bounded by a function of $b$. We state vertex bounds explicitly; a polynomial vertex bound also gives a polynomial encoding bound for these simple graphs.

The following kernel combines a domination-counting bound with a two-representative reduction for true-twin clique vertices. The individual ingredients are elementary; the contribution is their application to obtain a solution-size polynomial kernel on split graphs of bounded clique-to-independent degree.

\subsection{Counting Independent Vertices}

The first kernel rule is forced by capacity. A selected independent vertex accounts only for itself, while a selected clique vertex can dominate at most $d$ independent vertices. Hence a budget-$b$ solution cannot cover an arbitrarily large independent set.

\begin{lemma}\label{lem:budgetcount}
Let $G$ have split partition $(C,I)$ and $\Delta_I\leq d$, where $d\geq1$. Every dominating set $D$ satisfies
\[
 |I|\leq |D\cap I|+d|D\cap C|\leq d|D|.
\]
Consequently, $|I|>db$ certifies that $G$ has no CSDS of size at most $b$.
\end{lemma}
\begin{proof}
Every vertex of $I\setminus D$ must have a neighbour in $D\cap C$, since $I$ is independent. Each selected clique vertex has at most $d$ neighbours in $I$.
\end{proof}

\subsection{Reducing Equal Clique Neighbourhoods}

Vertices $x,y$ are \emph{true twins} if $N[x]=N[y]$. In a split partition, clique vertices are true twins exactly when they have the same neighbourhood in $I$. Two representatives suffice: one may be selected and the other may serve as a replacement.

\begin{lemma}[True-twin reduction]\label{lem:twins}
Let $X$ be a true-twin class of size at least three in a graph $G$. Let $G'$ be obtained by retaining two vertices of $X$ and deleting the others. Then $\gcs(G')=\gcs(G)$.
\end{lemma}
\begin{proof}
First, any CSDS $D$ can be changed, without increasing its size, so that $|D\cap X|\leq1$. If $|D\cap X|\geq2$, keep one selected vertex $x\in X$ and remove the others from $D$. The remaining set still dominates, because $x$ has the same closed neighbourhood as each removed vertex. For every selected vertex outside $X$, its original replacement remains valid: $x$ remains selected during that exchange and covers everything previously covered by the removed twins. The vertex $x$ can be replaced by one of the twins just removed from $D$.

Now take a minimum CSDS in this form. By symmetry, its selected vertex in $X$, if any, can be chosen among the two retained representatives. Any replacement in a deleted part of $X$ can be replaced by a retained unselected representative. Such a representative exists because at most one vertex of $X$ is selected. Restricting to $G'$ therefore gives a CSDS of no larger size.

Conversely, any dominating set of $G'$ dominates $G$: a deleted twin has the same intersection with that set in its closed neighbourhood as a retained twin. This applies both to a CSDS of $G'$ and to each of its replacement sets. Hence every CSDS of $G'$ is also a CSDS of $G$, proving equality.
\end{proof}

Keeping only one representative would not be safe. For example, a complete graph of order at least two has a singleton CSDS, but its reduction to one vertex has no replacement. The second representative preserves this possibility.

The two preceding lemmas control complementary parts of the instance: Lemma~\ref{lem:budgetcount} bounds $I$, and Lemma~\ref{lem:twins} bounds how many clique vertices need be retained for each neighbourhood in $I$. Counting the possible bounded-size neighbourhoods now gives the kernel size.

\begin{theorem}\label{thm:kernel}
For every fixed integer $d\geq1$, \CSDS\ on connected split graphs with $\Delta_I\leq d$ admits a kernel with at most
\begin{equation}\label{eq:kernelsize}
             db+2\sum_{j=0}^{d}\binom{db}{j}
\end{equation}
vertices, where $\binom{db}{j}=0$ for $j>db$. In particular, the vertex bound is $O(b^d)$ for fixed $d$.
\end{theorem}
\begin{proof}
The reduction has two stages. We first reject instances whose independent side exceeds the domination capacity of the budget. Otherwise, we keep two representatives of every clique-neighbourhood type and count how many such types can occur.

If $|I|>db$, return the fixed no-instance consisting of $K_2$ with budget zero. Otherwise, group the vertices of $C$ by their neighbourhoods in $I$ and retain at most two vertices from each group. Lemma~\ref{lem:twins} proves equivalence without changing the budget.

There are at most $\sum_{j=0}^{d}\binom{|I|}{j}$ neighbourhoods of size at most $d$. All independent vertices are retained, so the reduced graph has at most
\[
                  |I|+2\sum_{j=0}^{d}\binom{|I|}{j}
\]
vertices. Lemma~\ref{lem:budgetcount} gives~\eqref{eq:kernelsize}. The reduction only deletes clique twins. It preserves connectivity, the split partition, and the upper bound on $\Delta_I$. Grouping neighbourhoods and constructing the induced subgraph take polynomial time.
\end{proof}

\begin{corollary}\label{cor:star-kernel}
For every fixed $r\geq4$, \CSDS\ on $K_{1,r}$-free split graphs has an $O(b^{r-1})$-vertex kernel. In particular, $K_{1,4}$-free split graphs have a cubic kernel, while the subclass with $\Delta_I\leq2$ has a kernel with at most $4b^2+4b+2$ vertices.
\end{corollary}
\begin{proof}
A clique vertex with $r$ independent neighbours would centre an induced $K_{1,r}$, so $\Delta_I\leq r-1$. Apply Theorem~\ref{thm:kernel}. Since the reduced graph is induced, it remains $K_{1,r}$-free. Substituting $d=2$ in~\eqref{eq:kernelsize} gives the explicit quadratic bound.
\end{proof}

Enumerating subsets of size at most $b$ in the kernel and checking the replacement criterion gives an FPT algorithm. For fixed $d$, its running time is $2^{O(b\log(b+2))}(|V|+|E|)^{O(1)}$, with constants depending on $d$. More generally, the bound on the number of neighbourhood types gives running time
\[
       2^{O(db\log(db+2))}(|V|+|E|)^{O(1)},
\]
so the problem is FPT for the combined parameter $(b,d)$, and therefore for $(b,r)$ on $K_{1,r}$-free split graphs. We do not claim that~\eqref{eq:kernelsize} is polynomial in the combined parameter when $d$ is unbounded.

In contrast, parameterizing only by $r$ or only by $\Delta_I$ does not yield an XP algorithm unless $P=NP$: Theorem~\ref{thm:hardness} already gives NP-hardness at $\Delta_I=2$, and Corollary~\ref{cor:stars} gives it at $r=4$. Thus the solution budget, not the forbidden-star threshold alone, is essential to the positive parameterized result.

\section{Approximation Algorithms}\label{sec:approximation}

For the minimisation problem, a $\rho$-approximation returns a CSDS $D$ satisfying $|D|\leq\rho\gcs(G)$. The forbidden-star condition suggests a simple charging approach: if the algorithm constructs an independent CSDS, then no vertex can dominate too many of its members without creating a large induced star. We first exploit this observation for every $K_{1,r}$-free split graph, then use the auxiliary edge-cover structure to obtain a stronger guarantee on the incidence-split class used in the reduction.

\subsection{A Guarantee for the Full Forbidden-Star Class}

To apply the charging argument, we need a CSDS whose selected vertices are pairwise nonadjacent. A split partition makes such a set easy to construct, but the extra clique vertex needed in one case must also have a replacement. The following lemma handles both possibilities.

\begin{lemma}\label{lem:independent-construction}
Given a connected split graph of order at least two and a split partition $(C,I)$, an independent CSDS can be found in linear time.
\end{lemma}
\begin{proof}
If $I=\varnothing$, return any singleton clique vertex; another clique vertex is its replacement. Suppose $I\neq\varnothing$.

If every clique vertex has a neighbour in $I$, return $D=I$. This set dominates $G$. For each $u\in I$, any neighbour $c\in C$ replaces $u$: the vertex $c$ dominates the clique and $u$, while every other independent vertex remains selected.

Otherwise, choose $z\in C$ with $N_I(z)=\varnothing$ and return $D=I\cup\{z\}$. This set is independent and dominates $G$. A vertex $u\in I$ can be replaced by any of its clique neighbours. The vertex $z$ can be replaced by any vertex of $C\setminus\{z\}$, which is nonempty by connectivity. Scanning the split adjacency lists finds the required set and replacements in linear time.
\end{proof}

Because the constructed set is independent, an optimum vertex can cover at most $r-1$ of its members in a $K_{1,r}$-free graph. Charging each constructed vertex to a dominating vertex of an optimum solution gives the approximation factor.

\begin{theorem}\label{thm:star-approximation}
For every fixed $r\geq3$, a CSDS within a factor $r-1$ of optimum can be found in $O(|V|+|E|)$ time in a connected $K_{1,r}$-free split graph, given a split partition. In particular, the algorithm is a $3$-approximation on $K_{1,4}$-free split graphs.
\end{theorem}
\begin{proof}
We compare the constructed and optimum solutions through the closed neighbourhoods of optimum vertices. Independence gives the lower side of the count, and the forbidden-star condition gives the upper side.

Let $D$ be the independent CSDS from Lemma~\ref{lem:independent-construction}, and let $D^*$ be an optimum CSDS. For $x\in D$, independence gives $|N[x]\cap D|=1$. For $x\notin D$, the induced-star restriction gives $|N[x]\cap D|\leq r-1$. Since $D^*$ dominates $D$,
\[
              |D|\leq\sum_{x\in D^*}|N[x]\cap D|
                   \leq(r-1)|D^*|.
\]
The running time follows from Lemma~\ref{lem:independent-construction}.
\end{proof}

Theorem~\ref{thm:star-approximation} places the optimisation problem in APX for every fixed $r$. It does not establish APX-hardness. Corollary~\ref{cor:star-kernel} and Theorem~\ref{thm:star-approximation} complete claims (iii) and (iv) in the introduction's result summary.

\subsection{A Stronger Guarantee on the Incidence-Split Class}

Let $\mathcal{S}_2$ consist of the graphs $S(J)$ defined in Section~\ref{sec:hardness}, where $J$ is a nonempty simple graph of minimum degree at least two. The graph $J$ need not be connected. Every clique vertex of $S(J)$ has two independent neighbours, no two clique vertices have the same independent neighbourhood, and every independent vertex has degree at least two. The graphs in the hardness reduction belong to this class.

On this class, the generic independent-set construction ignores useful information. The edge-cover formulation instead suggests selecting a spanning forest of the auxiliary graph. This gives a solution with roughly one edge per auxiliary vertex, while Lemma~\ref{lem:counting} gives a lower bound of two edges per three vertices.

\begin{theorem}\label{thm:incidence-approximation}
On $\mathcal{S}_2$, a CSDS within a factor $3/2$ of optimum can be found in linear time, given the defining split partition.
\end{theorem}
\begin{proof}
The proof combines an explicit upper bound and the counting lower bound. A spanning forest supplies the feasible solution; normalisation and the edge-cover formulation allow the lower bound to be applied to an unrestricted optimum CSDS.

Let $G=S(J)$, and write $q=|V(J)|$, $m_J=|E(J)|$, and $\kappa$ for the number of components of $J$. Since $\delta(J)\geq2$, each component has at least three vertices and $m_J\geq q$.

Take one spanning tree in each component of $J$, and let $F_T$ be their union. The set $F_T$ is an edge cover with no $K_2$ component. Its size is $q-\kappa$, which is at least two and strictly less than $m_J$. Hence the corresponding set
\[
                         D_T=\{c_e:e\in F_T\}
\]
 is a CSDS of $G$ by Lemma~\ref{lem:edgecover}.

To compare with the unrestricted optimum, let $D^*$ be a minimum CSDS of $G$. The constructed solution gives
\[
                |D^*|\leq q-\kappa<\min\{q,m_J\}.
\]
Thus $D^*$ meets $C$: a dominating set contained in $I$ would have to equal $I$ and have size $q$. Lemma~\ref{lem:normalisation} applies and gives an optimum solution contained in $C$. This solution has at least two vertices because $q\geq3$ and each clique vertex covers only two independent vertices. Lemmas~\ref{lem:edgecover} and~\ref{lem:counting} therefore imply
\[
                            \gcs(G)\geq\frac{2q}{3}.
\]
Consequently,
\[
        \frac{|D_T|}{\gcs(G)}\leq
        \frac{q-\kappa}{2q/3}\leq\frac32.
\]
The auxiliary graph and its spanning forest can be constructed by scanning the split incidence lists. The total time is $O(|V(G)|+|E(G)|)$.
\end{proof}

This proves claim (v) in the introduction's result summary. The hypotheses defining $\mathcal{S}_2$ are necessary for this proof. If independent vertices have degree one, clique normalisation need not apply. If clique vertices share a two-element independent neighbourhood, an isolated edge may have a replacement through a twin, so Lemma~\ref{lem:edgecover} no longer characterises all clique solutions. We therefore do not claim a $3/2$-approximation for every split graph with $\Delta_I\leq2$.

\subsection{What the Hardness Result Does Not Imply}

The decision reduction distinguishes a solution of size $b$ from solutions of size at least $b+1$. The resulting relative gap $1+1/b$ tends to one. It therefore does not establish APX-hardness or exclude a polynomial-time approximation scheme (PTAS).

It does exclude a fully polynomial-time approximation scheme (FPTAS), unless $P=NP$. On an $n$-vertex graph, choose the accuracy $\varepsilon=1/(2n)$. Since $\gcs(G)\leq n$, an FPTAS would return an integer solution size less than $\gcs(G)+1$, and hence an optimum solution, in polynomial time. This would contradict the NP-hardness result. A PTAS does not give the same conclusion because its dependence on $1/\varepsilon$ need not be polynomial.

\section{Beyond Split Graphs: Stars and Peripheral Cliques}\label{sec:extensions}

The preceding sections restrict a split graph by excluding an induced star. Here we change the split structure itself. These are different operations: a $K_{1,r}$-split graph, defined below using attached stars, need not be a $K_{1,r}$-free split graph. All graphs in this section are connected, have a nonempty central clique $K$, and are supplied with the stated defining partition. The central clique need not be maximal. If there are no peripheral components, the graph is complete and its co-secure domination number is one, since its order is at least two. We henceforth assume at least one peripheral component.

Co-secure domination is already NP-complete on ordinary split graphs~\cite{KusumComplexity2023}. Thus hardness of a larger class by containment alone does not identify the effect of the added structure. For clique extensions, the relevant ambient class is the class of \emph{unipolar graphs}: a central clique together with an induced disjoint union of peripheral cliques~\cite{EschenWang2014}. We use explicit size and attachment restrictions to distinguish this setting from unrestricted unipolar graphs. The classifications below are proved directly; the cited literature supplies the ambient terminology and the source hardness results, rather than these attachment-specific conclusions.

\subsection{Pendant-Star Extensions}

For an integer $r\geq0$, a \emph{$K_{1,r}$-split partition} is a partition $V(G)=K\mathbin{\dot\cup}I\mathbin{\dot\cup}J$ with the following properties. The set $K$ is a clique, $I$ is a set of designated centres, and $J$ is partitioned into sets $J_u$, $u\in I$, where every vertex of $J_u$ is adjacent only to $u$. The centres are pairwise nonadjacent, every centre has at least one neighbour in $K$, and $0\leq p_u:=|J_u|\leq r$. A graph with such a partition is \emph{strict} if every $p_u=r$. We use \emph{non-strict} for the inclusive bounded-size class, which also contains strict instances. In particular, zero-leaf stars are allowed in the non-strict class, and $K_{1,0}$-split graphs are ordinary split graphs.

\begin{theorem}\label{thm:positive-pendant-stars}
Suppose every centre has at least one leaf. Put $I_1=\{u\in I:p_u=1\}$ and $d_1(c)=|N(c)\cap I_1|$ for $c\in K$. Then
\[
 \gcs(G)=
 \begin{cases}
 |J|,&d_1(c)\geq2\text{ for every }c\in K,\\
 |J|+1,&\text{otherwise}.
 \end{cases}
\]
A minimum CSDS and one replacement for each selected vertex can be found in $O(|V(G)|+|E(G)|)$ time.
\end{theorem}
\begin{proof}
If a centre $u$ is selected, none of its leaves can be selected: a selected leaf would have no neighbour outside the set. If $p_u\geq2$, removing $u$ would then leave at least two leaves undominated, and no single replacement could dominate both. Consequently every such centre is unselected and all its leaves are selected. If $p_u=1$, domination requires at least one of the centre and its leaf, and co-security forbids selecting both. Hence every CSDS contains exactly $|J|$ vertices outside $K$, plus its selected vertices of $K$.

A CSDS of size $|J|$ therefore contains no vertex of $K$. Its selected centres form a subset $S\subseteq I_1$ and must dominate $K$. The only possible replacement of $u\in S$ is its unique leaf: replacing $u$ by a clique vertex would leave that leaf undominated. Thus every clique vertex must have at least two neighbours in $S$, which requires $d_1(c)\geq2$ for every $c\in K$.

Conversely, if this condition holds, choose
\[
 D_0=I_1\cup\bigcup_{u\in I\setminus I_1}J_u.
\]
A selected centre in $I_1$ is replaced by its leaf, while all clique vertices remain dominated by other selected centres. A selected leaf of a centre outside $I_1$ is replaced by that centre. Each swap preserves domination, so $D_0$ is a CSDS of size $|J|$.

In all cases, fix $c\in K$ and choose $D_1=J\cup\{c\}$. Every selected leaf is replaced by its centre. If $|K|\geq2$, replace $c$ by any other clique vertex. If $K=\{c\}$, replace $c$ by any centre, which is adjacent to $c$ by assumption; all leaves remain selected. Hence $D_1$ is co-secure and has size $|J|+1$. Scanning the incidence lists computes the condition and constructs the appropriate set and replacements in linear time.
\end{proof}

\begin{corollary}\label{cor:pendant-star-classification}
Write $t=|I|$. On strict $K_{1,1}$-split graphs, the optimum is $t$ if every clique vertex has at least two centre neighbours, and $t+1$ otherwise. On strict $K_{1,r}$-split graphs with $r\geq2$, the optimum is $rt+1$. These strict positive-leaf cases are linear-time solvable. For every fixed $r\geq0$, \CSDS\ is NP-complete on the non-strict $K_{1,r}$-split class that allows zero-leaf stars.
\end{corollary}
\begin{proof}
The formulas follow from Theorem~\ref{thm:positive-pendant-stars}; for strict $r\geq2$, the set $I_1$ is empty. For non-strict hardness, take all peripheral stars to have zero leaves and use split-graph NP-completeness~\cite{KusumComplexity2023}. Membership in NP follows by checking domination and each proposed replacement. The same verification applies to the hardness results below.
\end{proof}

Thus unequal positive star sizes remain tractable. Allowing zero-leaf centres is the distinction needed for the containment argument; it should not be confused with merely dropping equality of positive star sizes.

\subsection{Peripheral Cliques and Their Attachment Sets}

Fix an integer $s\geq1$. A \emph{$K_s$-split partition} has the form
\[
 V(G)=K\mathbin{\dot\cup}B_1\mathbin{\dot\cup}\cdots
       \mathbin{\dot\cup}B_t,
\]
where each $B_i$ is a clique of size between one and $s$, there are no edges between distinct $B_i$, and edges between $K$ and the $B_i$ are otherwise arbitrary. The size bound $s$ is not the order of the input graph. The partition is \emph{strict} if every $|B_i|=s$; the non-strict class includes all sizes from one to $s$. We write $N_K(v)=N(v)\cap K$ and define the \emph{attachment set}
\[
 A_i=\{v\in B_i:N_K(v)\neq\varnothing\}.
\]
Connectivity implies $A_i\neq\varnothing$. A \emph{proper attachment set} satisfies $A_i\subsetneq B_i$, so the peripheral clique has at least one vertex with no central neighbour. Vertices in $A_i$ may have different central neighbourhoods; equality of these neighbourhoods will be a separate restriction. These definitions cover arbitrary subsets of peripheral vertices attaching to the main clique, without assuming that an attachment vertex is adjacent to all of $K$.

\begin{lemma}\label{lem:peripheral-additive}
Suppose every attachment set is proper. Then
\[
                         t\leq\gcs(G)\leq t+1.
\]
A CSDS of size $t+1$, with its replacements, can be constructed in linear time. Moreover, a set $D=\{x_1,\ldots,x_t\}$ with $x_i\in B_i$ is co-secure if and only if it dominates $K$ and, for each $i$, some $y_i\in B_i\setminus\{x_i\}$ satisfies
\begin{equation}\label{eq:peripheral-transversal}
 \{c\in K:N(c)\cap D=\{x_i\}\}\subseteq N_K(y_i).
\end{equation}
\end{lemma}
\begin{proof}
A vertex in $B_i\setminus A_i$ can only be dominated from $B_i$, so every dominating set meets all $t$ peripheral cliques. Choose one non-attachment vertex $x_i\in B_i\setminus A_i$ for each $i$ and one $c\in K$. The set $D^+=\{x_1,\ldots,x_t,c\}$ dominates $G$. Each selected $x_i$ can be replaced by any other vertex of $B_i$, since $c$ continues to dominate $K$. If $|K|\geq2$, replace $c$ by another clique vertex. If $K=\{c\}$, replace $c$ by any $a\in A_1$. This vertex is unselected and adjacent to $c$, and the retained $x_i$ continue to dominate all peripheral cliques. This proves the upper bound and the construction time.

Now consider a set with exactly one selected vertex $x_i$ in each $B_i$ and none in $K$. It dominates all peripheral cliques. Any valid replacement of $x_i$ must lie in $B_i$: a vertex of $B_i\setminus A_i$ would otherwise lose domination, including when it is $x_i$ itself. After a replacement inside $B_i$, the only possible losses of domination are the clique vertices whose unique selected neighbour was $x_i$. This gives precisely~\eqref{eq:peripheral-transversal}.
\end{proof}

\subsection{One Attachment Vertex and Equal Attachment Neighbourhoods}

\begin{theorem}\label{thm:uniform-attachments}
Suppose every $A_i$ is proper and all vertices of $A_i$ have the same nonempty central neighbourhood $T_i\subseteq K$. For $c\in K$, let $\mu(c)=|\{i:c\in T_i\}|$. Then $\gcs(G)=t$ exactly when every $c\in K$ satisfies
\begin{equation}\label{eq:uniform-attachment-condition}
 \mu(c)\geq2
 \quad\text{or}\quad
 c\in T_i\text{ for some }i\text{ with }|A_i|\geq2.
\end{equation}
Otherwise $\gcs(G)=t+1$. A minimum CSDS and its replacements can be found in linear time, given the partition.
\end{theorem}
\begin{proof}
The lower and upper bounds follow from Lemma~\ref{lem:peripheral-additive}. If $\mu(c)=0$, a set of size $t$ cannot dominate $c$, since such a set contains no central vertex. If $\mu(c)=1$ and the unique relevant attachment set is $A_i=\{u_i\}$, any size-$t$ dominating set must select $u_i$. Replacing it outside $B_i$ leaves a non-attachment vertex undominated, whereas replacing it inside $B_i$ leaves $c$ undominated. Thus condition~\eqref{eq:uniform-attachment-condition} is necessary.

If the condition holds, select one $x_i\in A_i$ from each peripheral clique. When $|A_i|\geq2$, replace $x_i$ by another member of $A_i$, which has the same closed neighbourhood. When $A_i=\{x_i\}$, replace $x_i$ by any vertex of $B_i\setminus A_i$. Every central neighbour of $x_i$ belongs to at least two of the sets $T_j$ by the condition, and therefore remains dominated. The selected vertices form a CSDS of size $t$. The sets $T_i$ can be read from one attachment vertex per peripheral clique. A scan counts $\mu(c)$ and records whether $c$ belongs to a $T_i$ with $|A_i|\geq2$, giving linear time.
\end{proof}

\begin{corollary}\label{cor:single-attachment}
Suppose every peripheral clique has at least two vertices and exactly one attachment vertex $u_i$. Put $U=\{u_1,\ldots,u_t\}$. Then
\[
 \gcs(G)=
 \begin{cases}
 t,&|N(c)\cap U|\geq2\text{ for every }c\in K,\\
 t+1,&\text{otherwise}.
 \end{cases}
\]
This gives a linear-time exact algorithm for strict $K_s$-split graphs with one attachment vertex per peripheral clique for every $s\geq2$, and also for unequal peripheral sizes all at least two. If singleton peripheral cliques are permitted, the class is NP-complete by containment of ordinary split graphs.
\end{corollary}
\begin{proof}
Take $A_i=\{u_i\}$ in Theorem~\ref{thm:uniform-attachments}. For the last statement, a singleton peripheral vertex acts as its own attachment vertex, and choosing all peripheral cliques to be singletons gives an arbitrary connected split graph~\cite{KusumComplexity2023}.
\end{proof}

For $s=2$, the single-attachment model coincides with strict $K_{1,1}$-split graphs. With several attachment vertices, the equality of their central neighbourhoods is important: the next result shows that the tractability conclusion does not extend to arbitrary neighbourhoods, even when there are only two attachment vertices per peripheral clique.

\subsection{Two Attachment Vertices Already Give Hardness}

\begin{theorem}\label{thm:two-attachment-hardness}
For every fixed integer $s\geq3$, \CSDS\ is NP-complete on strict $K_s$-split graphs with exactly two attachment vertices in every peripheral clique and no central neighbours for the other $s-2$ vertices. The defining partition may be supplied. Hardness holds for deciding whether $\gcs(G)=t$, where $t$ is the number of peripheral cliques and $\gcs(G)\in\{t,t+1\}$.
\end{theorem}
\begin{proof}
We reduce from Boolean 3-SAT, the standard NP-complete restriction of Boolean satisfiability~\cite{Cook1971}. Let $\Phi$ have variables $z_1,\ldots,z_q$ and clauses $\mathcal C_1,\ldots,\mathcal C_m$, each containing exactly three literals. A positive literal $z_i$ is satisfied when $z_i$ is true, and a negative literal $\neg z_i$ when $z_i$ is false. We may assume $q\geq1$ by adding an unused variable.

For every variable, create a peripheral clique $B_i$ of order $s$ with distinguished vertices $x_i^+,x_i^-$. Create two further peripheral cliques $B_a,B_b$ of order $s$ with distinguished vertices $a^+,a^-$ and $b^+,b^-$. There are $t=q+2$ peripheral cliques. All vertices not distinguished have no neighbours outside their own peripheral clique.

The central clique is
\[
 K=\{h_1,\ldots,h_q,h_a,h_b,w,c_1,\ldots,c_m\}.
\]
Besides the clique edges, add exactly the following edges:
\begin{itemize}[nosep]
\item $h_i$ is adjacent to $x_i^+$ and $x_i^-$ for $1\leq i\leq q$; $h_a$ is adjacent to $a^+,a^-$, and $h_b$ to $b^+,b^-$.
\item $w$ is adjacent to $a^+$ and $b^+$.
\item For each clause $\mathcal C_j$, the vertex $c_j$ is adjacent to $a^+$, to every $x_i^+$ with $z_i\in\mathcal C_j$, and to every $x_i^-$ with $\neg z_i\in\mathcal C_j$.
\end{itemize}
Each distinguished vertex has a central neighbour through its corresponding $h$ vertex. Hence every peripheral clique has exactly two attachment vertices and $s-2\geq1$ non-attachment vertices. The graph is connected and has $(s+1)(q+2)+m+1$ vertices, so its construction is polynomial. Set the decision budget to $t=q+2$.

Figure~\ref{fig:two-attachment-reduction} illustrates the construction for a small satisfiable formula and shows the size-$t$ CSDS together with its replacements.

\begin{figure}[t]
\centering
\resizebox{\linewidth}{!}{%
\begin{tikzpicture}[
  font=\small,
  vertex/.style={circle,fill=black,inner sep=0pt,minimum size=4pt},
  selected/.style={vertex},
  replacement/.style={vertex},
  every label/.style={font=\normalsize,inner sep=2pt},
  group/.style={draw=black!45,rounded corners,fill=gray!5,thick}
]
\node[font=\small\bfseries,anchor=west] at (0,5.95) {Constructed graph $G_\Phi$};
\node[font=\scriptsize,anchor=east] at (13.7,5.95) {$K$ is a clique (internal edges omitted)};

\begin{scope}[on background layer]
  \draw[group,dashed] (-0.15,2.65) rectangle (13.85,5.4);
  \draw[group] (0,-0.55) rectangle (2.4,1.45);
  \draw[group] (2.8,-0.55) rectangle (5.2,1.45);
  \draw[group] (5.6,-0.55) rectangle (8.0,1.45);
  \draw[group] (8.4,-0.55) rectangle (10.8,1.45);
  \draw[group] (11.2,-0.55) rectangle (13.6,1.45);
\end{scope}

\node[vertex,label=above:$h_1$] (h1) at (1.2,4.75) {};
\node[vertex,label=above:$h_2$] (h2) at (4.0,4.75) {};
\node[vertex,label=above:$h_3$] (h3) at (6.8,4.75) {};
\node[vertex,label=above:$h_a$] (ha) at (9.6,4.75) {};
\node[vertex,label=above:$h_b$] (hb) at (12.4,4.75) {};
\node[vertex,label=above:$c_1$] (c1) at (2.0,3.55) {};
\node[vertex,label=above:$c_2$] (c2) at (4.9,3.55) {};
\node[vertex,label=above:$c_3$] (c3) at (7.8,3.55) {};
\node[vertex,label=above:$c_4$] (c4) at (10.5,3.55) {};
\node[vertex,label=above:$w$] (w) at (11.2,3.0) {};

\node[font=\scriptsize] at (1.2,1.23) {$B_1$};
\node[font=\scriptsize] at (4.0,1.23) {$B_2$};
\node[font=\scriptsize] at (6.8,1.23) {$B_3$};
\node[font=\scriptsize] at (9.6,1.23) {$B_a$};
\node[font=\scriptsize] at (12.4,1.23) {$B_b$};

\node[selected,label={[text=csdsBlue]left:$x_1^+$}] (x1p) at (0.78,0.65) {};
\node[replacement,label={[text=csdsOrange]right:$x_1^-$}] (x1m) at (1.62,0.65) {};
\node[vertex,label=below:$p_1$] (p1) at (1.2,-0.1) {};
\node[selected,label={[text=csdsBlue]left:$x_2^+$}] (x2p) at (3.58,0.65) {};
\node[replacement,label={[text=csdsOrange]right:$x_2^-$}] (x2m) at (4.42,0.65) {};
\node[vertex,label=below:$p_2$] (p2) at (4.0,-0.1) {};
\node[selected,label={[text=csdsBlue]left:$x_3^+$}] (x3p) at (6.38,0.65) {};
\node[replacement,label={[text=csdsOrange]right:$x_3^-$}] (x3m) at (7.22,0.65) {};
\node[vertex,label=below:$p_3$] (p3) at (6.8,-0.1) {};
\node[selected,label={[text=csdsBlue]left:$a^+$}] (ap) at (9.18,0.65) {};
\node[replacement,label={[text=csdsOrange]right:$a^-$}] (am) at (10.02,0.65) {};
\node[vertex,label=below:$p_a$] (pa) at (9.6,-0.1) {};
\node[selected,label={[text=csdsBlue]left:$b^+$}] (bp) at (11.98,0.65) {};
\node[replacement,label={[text=csdsOrange]right:$b^-$}] (bm) at (12.82,0.65) {};
\node[vertex,label=below:$p_b$] (pb) at (12.4,-0.1) {};

\begin{scope}[on background layer]
  \foreach \u/\v in {x1p/x1m,x1p/p1,x1m/p1,x2p/x2m,x2p/p2,x2m/p2,x3p/x3m,x3p/p3,x3m/p3,ap/am,ap/pa,am/pa,bp/bm,bp/pb,bm/pb}
    \draw[black!45] (\u)--(\v);
  \foreach \u/\v in {h1/x1p,h1/x1m,h2/x2p,h2/x2m,h3/x3p,h3/x3m,ha/ap,ha/am,hb/bp,hb/bm,w/ap,w/bp}
    \draw[black!65] (\u)--(\v);
  \foreach \u/\v in {c1/ap,c1/x1p,c1/x2m,c1/x3p,c2/ap,c2/x1m,c2/x2p,c2/x3p,c3/ap,c3/x1p,c3/x2p,c3/x3m,c4/ap,c4/x1m,c4/x2m,c4/x3p}
    \draw[csdsBlue!45] (\u)--(\v);
  \foreach \u/\v in {x1p/x1m,x2p/x2m,x3p/x3m,ap/am,bp/bm}
    \draw[csdsOrange,line width=1.2pt] (\u)--(\v);
\end{scope}

\node[font=\scriptsize,align=center] at (7,-1.05)
  {$D=\{x_1^+,x_2^+,x_3^+,a^+,b^+\}$ (blue labels); orange labels and edges mark replacements\\
   for the paired selected vertices};
\end{tikzpicture}%
}
\caption{Graph constructed by the two-attachment reduction for the 3-CNF formula $\Phi=(z_1\vee\neg z_2\vee z_3)\wedge(\neg z_1\vee z_2\vee z_3)\wedge(z_1\vee z_2\vee\neg z_3)\wedge(\neg z_1\vee\neg z_2\vee z_3)$, satisfied by $z_1=z_2=z_3=\mathrm{true}$. Here $s=3$ and $t=q+2=5$. The graph is a strict $K_3$-split graph with a size-$t$ CSDS. Blue labels identify the CSDS and orange labels and edges identify replacements. All edges within the central clique are suppressed. Colored edges are for better visualization. For $s>3$, add the extra non-attachment vertices to each peripheral clique.}
\label{fig:two-attachment-reduction}
\end{figure}

Suppose first that $\Phi$ is satisfiable. Select $a^+$, $b^+$, and $x_i^+$ or $x_i^-$ according to a satisfying truth assignment. The resulting set $D$ has size $t$ and dominates every peripheral clique. It dominates each $h$ vertex through its own peripheral clique, $w$ through both $a^+,b^+$, and every clause vertex through $a^+$ and at least one selected literal vertex.

Replace a selected literal vertex by the opposite distinguished vertex in the same peripheral clique. Its $h$ vertex remains dominated, and all clause vertices remain dominated by $a^+$. Replace $a^+$ by $a^-$: the vertex $w$ remains dominated by $b^+$, and each clause vertex retains a selected satisfying literal. Finally, $b^-$ replaces $b^+$, while $a^+$ continues to dominate $w$. Thus $D$ is co-secure.

Conversely, let $D$ be a CSDS of size at most $t$. Every peripheral clique contains a non-attachment vertex, so Lemma~\ref{lem:peripheral-additive} forces exactly one selected vertex per peripheral clique and none in $K$. To dominate each $h$ vertex, that selected vertex must be one of the two distinguished vertices. Moreover, any replacement must stay in its own peripheral clique, by the same lemma.

The vertex $w$ must have both $a^+$ and $b^+$ selected. If neither is selected, it is undominated. If exactly one is selected, removing that vertex leaves $w$ undominated, since no other vertex of the same peripheral clique is adjacent to $w$.

Interpret selection of $x_i^+$ as assigning true to $z_i$ and selection of $x_i^-$ as assigning false. If some clause has no selected satisfying literal, its clause vertex has $a^+$ as its unique selected neighbour. No vertex of $B_a\setminus\{a^+\}$ is adjacent to that clause vertex. Thus $a^+$ has no valid replacement, a contradiction. The assignment satisfies $\Phi$.

We have proved that $\Phi$ is satisfiable if and only if $\gcs(G)\leq t$. Lemma~\ref{lem:peripheral-additive} restricts the optimum to $t$ or $t+1$, and membership in NP completes the proof.
\end{proof}

Therefore strict $K_2$-split graphs with proper attachment sets have the single-attachment formula, whereas strict $K_s$-split graphs with proper attachment sets are NP-complete for every fixed $s\geq3$. The same hardness holds for non-strict peripheral sizes between two and $s$ when $s\geq3$, even though singleton peripheral cliques are excluded. This is not a containment argument from ordinary split graphs.

\begin{corollary}\label{cor:fixed-attachment-count}
For every pair of fixed integers $s,\ell$ with $2\leq\ell<s$, \CSDS\ is NP-complete on strict $K_s$-split graphs with exactly $\ell$ attachment vertices in each peripheral clique.
\end{corollary}
\begin{proof}
In the construction of Theorem~\ref{thm:two-attachment-hardness}, choose $\ell-2$ of the non-attachment vertices of each peripheral clique and join them only to its own central $h$ vertex. There remain $s-\ell\geq1$ non-attachment vertices per peripheral clique. A satisfying assignment still gives the same CSDS, with the same replacements. Conversely, a size-$t$ CSDS still selects exactly one vertex from each peripheral clique, and the vertex $w$ still forces selection of $a^+$ and $b^+$. The added attachment vertices have no clause neighbours, so every clause must retain a selected literal vertex when $a^+$ is removed. Read the selected literal vertices as a partial truth assignment, and assign arbitrary truth values to variables whose peripheral clique selects an added attachment vertex. Every clause already has a selected satisfying literal, so this extension satisfies the formula. The reduction remains polynomial.
\end{proof}

\begin{corollary}\label{cor:partial-attachment-approximation}
For every fixed $s\geq2$, graphs with proper attachment sets and $2\leq|B_i|\leq s$ admit a linear-time additive-one approximation and a polynomial-time approximation scheme. They also admit an exact algorithm with running time $s^t$ times a polynomial in the input size. For $s\geq3$, an FPTAS would imply $P=NP$.
\end{corollary}
\begin{proof}
Lemma~\ref{lem:peripheral-additive} constructs a set of size $t+1$ against the lower bound $t$. A solution of size $t$, if one exists, selects one vertex from each peripheral clique. There are at most $s^t$ such choices, each checkable in polynomial time using~\eqref{eq:peripheral-transversal}. Enumerating them either finds an optimum of size $t$ or certifies that the constructed size-$t+1$ solution is optimum.

Given $\varepsilon>0$, return the constructed solution when $t\geq1/\varepsilon$, since $(t+1)/t\leq1+\varepsilon$. Otherwise enumerate the choices exactly. The running time is at most $s^{\lceil1/\varepsilon\rceil}$ times a polynomial in the input size, and is polynomial for fixed $s,\varepsilon$. For an FPTAS, choosing $\varepsilon=1/(2(t+1))$ would yield an integer solution size strictly less than $\gcs(G)+1$, since $\gcs(G)\leq t+1$. It would therefore find an optimum in polynomial time, contradicting Theorem~\ref{thm:two-attachment-hardness} unless $P=NP$.
\end{proof}

The PTAS here concerns proper peripheral attachments. It does not settle the PTAS question for the incidence-split class of Section~\ref{sec:approximation}, whose peripheral components are singletons and have no non-attachment vertices.

\subsection{Unrestricted and Full Peripheral Attachments}

When $A_i=B_i$ is permitted, the lower bound of one selected vertex per peripheral clique need not hold. A separate reduction handles strict equal-size peripheral cliques, including the case where all their vertices attach with identical neighbourhoods.

\begin{theorem}\label{thm:unrestricted-clique-extension}
For every fixed integer $s\geq1$, \CSDS\ is NP-complete on strict and on non-strict $K_s$-split graphs with unrestricted central--peripheral edges. For $s\geq2$, strict hardness holds even when $A_i=B_i$ and every vertex of $B_i$ has the same central neighbourhood.
\end{theorem}
\begin{proof}
For $s=1$, use co-secure domination on split graphs~\cite{KusumComplexity2023}. Fix $s\geq2$. Ordinary \textsc{Dominating Set} is NP-complete on split graphs~\cite{Bertossi1984}. This remains a suitable source when the graph is connected: any isolated independent vertices can be deleted and their number subtracted from the budget, leaving a connected nonempty split graph, apart from directly decidable cases. Let $H$ be such a source graph with partition $(C,I)$, and write $\gamma(H)$ for its ordinary domination number.

Replace every vertex $v\in C$ by a clique $X_v$ of order two, and every $v\in I$ by a clique $X_v$ of order $s$. Join every vertex of $X_u$ to every vertex of $X_v$ exactly when $uv\in E(H)$. The union of the $X_v$ for $v\in C$ is the central clique of the resulting graph $G_s$, and the $X_v$ for $v\in I$ are its peripheral cliques. All vertices of each replacement clique have identical closed neighbourhoods. We claim that
\begin{equation}\label{eq:clique-expansion-optimum}
                         \gcs(G_s)=\gamma(H).
\end{equation}

Given a minimum dominating set of $H$, choose one vertex from each corresponding replacement clique. The chosen vertices dominate $G_s$. Every chosen vertex has an unselected true twin in its replacement clique, because each such clique has at least two vertices and only one was chosen. Swapping these twins preserves domination, proving $\gcs(G_s)\leq\gamma(H)$.

Conversely, for any dominating set $D$ of $G_s$, let $S=\{v\in V(H):D\cap X_v\neq\varnothing\}$. Every original vertex lies in $S$ or is adjacent to a vertex in $S$, since otherwise its replacement clique would be undominated. Thus $S$ dominates $H$ and $|S|\leq|D|$. This also applies to co-secure dominating sets and proves the reverse inequality in~\eqref{eq:clique-expansion-optimum}. The construction preserves the budget and is polynomial for fixed $s$.

Each independent vertex of connected $H$ has a clique neighbour, so all vertices of every peripheral clique have a nonempty, identical central neighbourhood. Strict hardness follows. Non-strict hardness follows either by including the strict class or by allowing all peripheral sizes to be one. Membership in NP is immediate.
\end{proof}

Table~\ref{tab:extension-classification} summarizes the classifications. The difference between proper and full attachment is essential. Proper attachment forces every solution to meet every peripheral clique. Full attachment can instead allow a selected central vertex to dominate several entire peripheral cliques, recovering the covering choices responsible for split-graph hardness. With proper attachment, Theorems~\ref{thm:uniform-attachments} and~\ref{thm:two-attachment-hardness} further distinguish equal from unequal attachment neighbourhoods.

\begin{table}[htbp]
\centering
\caption{Co-secure domination on the extensions defined in Section~\ref{sec:extensions}. All size bounds are fixed; linear-time algorithms assume the defining partition is supplied. Non-strict classes include the stated smallest peripheral sizes.}
\label{tab:extension-classification}
\begin{tabular}{p{0.68\linewidth}p{0.22\linewidth}}
\toprule
Class or attachment restriction & Complexity\\
\midrule
Strict pendant $K_{1,r}$, $r\geq1$ & Linear time\\
Non-strict pendant $K_{1,r}$, allowing zero leaves & NP-complete\\
Unequal pendant stars, all with positive leaf counts & Linear time\\
Strict or non-strict $K_s$-split, unrestricted attachments & NP-complete\\
Peripheral sizes at least two, exactly one attachment each & Linear time\\
Proper attachments, equal central neighbourhoods within each attachment set & Linear time\\
Strict $K_s$-split, exactly $\ell$ attachments each, $2\leq\ell<s$ & NP-complete\\
Proper attachments, sizes between two and fixed $s$ & PTAS\\
\bottomrule
\end{tabular}
\end{table}

\section{Conclusion and Further Directions}\label{sec:conclusion}

We identify a complexity boundary for co-secure domination on split graphs. Claw-free split graphs admit a linear-time exact algorithm, while the problem is NP-complete on $K_{1,r}$-free split graphs for every fixed $r\geq4$. These hard classes still admit $O(b^{r-1})$-vertex kernels and linear-time $(r-1)$-approximations; a quadratic kernel and a $3/2$-approximation hold for the specified subclasses. The extensions show that peripheral structure matters: positive-leaf pendant-star graphs are tractable, while unrestricted peripheral-clique attachments are hard. Proper attachments admit linear-time algorithms when each clique has one attachment vertex or its attachment vertices have identical central neighbourhoods. Two attachments with unequal neighbourhoods already give hardness for triangles, while bounded peripheral size admits a PTAS.

Several questions remain. Can the cubic kernel for $K_{1,4}$-free split graphs be reduced to a quadratic kernel? Can the $3/2$-approximation be extended to graphs with repeated clique neighbourhoods or independent vertices of degree one? Which restrictions on unequal attachment neighbourhoods preserve exact solvability? The PTAS for proper peripheral attachments does not settle approximation on the incidence-split class.

\section*{Declarations}

\textbf{Funding:} The authors received no external funding for this work.\\[2mm]
\textbf{Conflicts of Interest:} The authors have no competing interests to declare.\\[2mm]
\textbf{Data Availability:} No datasets were generated or analyzed during the current study.\\[2mm]
\textbf{Use of generative AI:} An OpenAI AI assistant was used during manuscript preparation to assist with language editing, LaTeX formatting, organization, and checking of mathematical arguments. The authors retain full responsibility for the correctness, originality, and integrity of the manuscript, including its proofs and references.


\begin{thebibliography}{99}

\bibitem{co-secureDS}
S. Arumugam, K. Ebadi, and M. Manrique, Co-secure and secure domination in graphs, Util. Math., 94 (2014), 167--182.

\bibitem{Bertossi1984}
A. A. Bertossi, Dominating sets for split and bipartite graphs, Inf. Process. Lett., 19 (1984), 37--40.
\doi{10.1016/0020-0190(84)90126-1}.

\bibitem{Cook1971}
S. A. Cook, The complexity of theorem-proving procedures, in Proceedings of the Third Annual ACM Symposium on Theory of Computing, ACM (1971), 151--158.
\doi{10.1145/800157.805047}.

\bibitem{EschenWang2014}
E. M. Eschen and X. Wang, Algorithms for unipolar and generalized split graphs, Discrete Appl. Math., 162 (2014), 195--201.
\doi{10.1016/j.dam.2013.08.011}.

\bibitem{KarthikaFCT}
D. Karthika, R. Muthucumaraswamy, V. P. Abidha, P. Ashok, S. Bhyravarapu, S. Das, S. Saurabh, A. Sawlani, and V. Tripathi, On the parameterized complexity of cosecure domination, in FCT 2025, Lecture Notes in Computer Science, vol. 16106, Springer (2026), 281--294.
\doi{10.1007/978-3-032-04700-7_21}.

\bibitem{KusumComplexity2023}
Kusum and A. Pandey, Complexity results on cosecure domination in graphs, in CALDAM 2023, Lecture Notes in Computer Science, vol. 13947, Springer (2023), 335--347.
\doi{10.1007/978-3-031-25211-2_26}.

\bibitem{CSDD_split_chordalbipartite_convex}
Kusum and A. Pandey, Some new algorithmic results on co-secure domination in graphs, Theor. Comput. Sci., 992 (2024), 114451.
\doi{10.1016/j.tcs.2024.114451}.

\bibitem{MonnotToulouse2007}
J. Monnot and S. Toulouse, The path partition problem and related problems in bipartite graphs, Oper. Res. Lett., 35 (2007), 677--684.
\doi{10.1016/j.orl.2006.12.004}.

\bibitem{Panda2024}
B. S. Panda, S. Rana, and S. Mishra, On the complexity of co-secure dominating set problem, Inf. Process. Lett., 185 (2024), 106463.
\doi{10.1016/j.ipl.2023.106463}.

\bibitem{PandaCorrigendum2026}
B. S. Panda, S. Rana, and S. Mishra, Corrigendum to ``On the complexity of co-secure dominating set problem'' [Inf. Process. Lett., 185 (2024), 106463], Inf. Process. Lett., 191 (2026), 106592.
\doi{10.1016/j.ipl.2025.106592}.

\bibitem{RENJITH2020246}
P. Renjith and N. Sadagopan, The Steiner tree in $K_{1,r}$-free split graphs---A dichotomy, Discrete Appl. Math., 280 (2020), 246--255.
\doi{10.1016/j.dam.2018.05.050}.

\bibitem{Renjith2026PID}
P. Renjith, M. M. Sathi, and M. V. Malayil, Perfect Italian domination on structured bipartite graphs: algorithms, parameterization, and hardness, Mathematics, 14(18) (2026), 3380.
\doi{10.3390/math14183380}.

\bibitem{Wang2025}
C.-X. Wang, Y. Yang, and S.-J. Xu, The algorithm and complexity of co-secure domination in geometric intersection graphs, Comput. Appl. Math., 44 (2025), 22.
\doi{10.1007/s40314-024-02982-2}.

\bibitem{west2001introduction}
D. B. West, Introduction to graph theory, 2nd ed., Prentice Hall (2001).

\bibitem{ZouProperInterval2022}
Y.-H. Zou, J.-J. Liu, S.-C. Chang, and C.-C. Hsu, The co-secure domination in proper interval graphs, Discrete Appl. Math., 311 (2022), 68--71.
\doi{10.1016/j.dam.2022.01.013}.

\end{thebibliography}
\end{document}